\documentclass[conference]{IEEEtran}

\usepackage{graphicx} 
\usepackage{comment}
\usepackage{amssymb}
\usepackage{amsmath}
\usepackage[hidelinks]{hyperref}
\usepackage{pbalance}
\usepackage[ruled,vlined,linesnumbered]{algorithm2e}
\usepackage{array}
\usepackage{refcount}

\usepackage{amsthm}
\newtheorem{theorem}{Theorem}
\newtheorem{lemma}{Lemma}
\newtheorem{definition}{Definition}

\makeatletter
\newcommand{\removelatexerror}{\let\@latex@error\@gobble}
\makeatother

\usepackage{graphicx} 

\IEEEoverridecommandlockouts 
\begin{document}
\title{\emph{TRAIL}: Cross-Shard Validation for \\ Cryptocurrency Byzantine Shard Protection\thanks{The paper is eligible for the best student paper award.}}

\author{
Mitch Jacovetty, Joseph Oglio, Mikhail Nesterenko, and Gokarna Sharma\\
\textit{Department of Computer Science, Kent State University,
Kent, OH 44242, USA} \\
mjacovet@kent.edu, joglio@kent.edu, mikhail@cs.kent.edu, sharma@cs.kent.edu}

\date{}

\sloppy
\maketitle
\thispagestyle{plain}
\pagestyle{plain}

\begin{abstract}
We present \emph{TRAIL}: an algorithm that uses a novel consensus procedure to tolerate failed or malicious shards within a blockchain-based cryptocurrency. Our algorithm takes a new approach of selecting validator shards for each transaction from those that previously held the assets being transferred. This approach ensures the algorithm's robustness and efficiency. \emph{TRAIL} is presented using \emph{PBFT} for internal shard transaction processing and a modified version of \emph{PBFT} for external cross-shard validation. We describe \emph{TRAIL}, prove it correct, analyze its message complexity, and evaluate its performance. We propose various \emph{TRAIL} optimizations: we describe how it can be adapted to other Byzantine-tolerant consensus algorithms, how a complete system may be built on the basis of it, and how \emph{TRAIL} can be applied to existing and future sharded blockchains.
\end{abstract}


\begin{IEEEkeywords}
Blockchain, Cryptocurrency, Sharding, Byzantine faults, Cross-shard validation, Distributed consensus.
\end{IEEEkeywords}


\section{Introduction}

In this paper, we present \emph{TRAIL} -- an algorithm for robust cryptocurrency blockchain design. A blockchain is a shared, immutable, append-only distributed ledger. 
A cryptocurrency blockchain~\cite{nakamoto,ethereum} is typically maintained by a peer-to-peer network. This design eliminates centralized control over transaction processing and makes the system potentially more scalable, flexible, and efficient.

To function as money, cryptocurrency needs to withstand network failures and malicious user behavior. It is usually designed to tolerate Byzantine faults~\cite{lamport1982byzantine}. A Byzantine peer may deviate from the algorithm and behave arbitrarily.
Therefore, such faults encompass a variety of failures and security threats. Despite the faults, correct peers need to be able to arrive at consensus on proposed transactions. 

Popular cryptocurrencies use proof-of-work based consensus algorithms~\cite{nakamoto}
in which peers compete for the right to publish records on the blockchain by searching for solutions to cryptographic challenges. Such algorithms tend to be conceptually simple and robust. However, they are resource intensive and environmentally harmful~\cite{wendl2023environmental}. Therefore, modern blockchain designs often focus on cooperative consensus algorithms.

In these cooperative consensus algorithms, 
rather than compete, peers exchange messages to arrive at a joint decision. Such algorithms may tolerate some number $f$ of faulty processes. This number is called tolerance threshold. It is usually a fraction of the network size $n$. 
One of the most widely used algorithms in this category is \emph{PBFT}~\cite{pbft}. 

Scaling up a Byzantine-robust algorithm is challenging as it usually involves system-wide broadcasts. Such broadcasts are expensive in large systems. A prominent approach of improving scalability in blockchains is sharding. In sharding, the network peers are divided into committees or shards. Each shard is made responsible for a subset of the processing done or the data stored by the network. Every shard internally runs a consensus algorithm, such as \emph{PBFT}, and coordinates with other shards to achieve global consistency. Thus, the overall workload is distributed and the processing of records is potentially accelerated. 

However, such sharding is at cross-purposes with fault tolerance: the network is only as reliable as any of its shards. For example, given a fixed number of peers, decreasing the shard size increases the number of available shards. This results in greater parallelism in transaction processing. Yet, a small shard is more vulnerable to failure since it has lower tolerance threshold $f$ of its internal consensus algorithm. 

The sharded blockchains presented in the literature usually assume that no shard tolerance threshold is breached. This places a limit on the efficiency of the sharding approach to performance improvement since shards need to be made large enough to ensure that they never fail. 

In this paper, we address the handling of complete shard failures which potentially allows aggressively small shards and removes the shard size scalability obstacle. A naive approach would be to group shards into static meta-shards. Such a meta-shard would treat individual shards as peers and run a meta-consensus algorithm among them to validate transactions across shards to withstand individual shard failures. However, concurrent transactions that are assigned to different shards would be verified by the same static meta-shard, regardless of the transactions' nature or history. This may create a performance bottleneck. 

\ \\
\textbf{Paper contribution.}  We propose \emph{TRAIL}: a 
novel application-specific approach to cross-shard validation. With this technique, a trail of shards dynamically follows each coin according its transaction history movement. The source shard runs an internal shard consensus algorithm to validate and linearize transactions. The trail of shards runs a cross-shard consensus algorithm to confirm the transaction and fortify it against shard failure. We present \emph{TRAIL} using \emph{PBFT} for both internal shard transaction processing and external cross-shard validation. We utilize \emph{PBFT} since it is well-known and widely used. 
Our solution may use various \emph{PBFT} efficiency enhancements such as parallel transaction processing and transaction pipelining.
Moreover, \emph{TRAIL} is independent of the specifics of sharding operation and may be adapted to enhance the robustness of consensus algorithms other than \emph{PBFT}.

We evaluate the performance of \emph{TRAIL} using an abstract simulator and study its transaction  confirmation rate, scalability and robustness against peer and shard failure. Our experiments indicate that \emph{TRAIL} adds shard failure protection with relatively modest resource expenditure. 

%
%

\ \\
\textbf{Paper organization.} 
We describe the network model and other preliminaries in Section~\ref{secPreliminaries}. We present \emph{TRAIL} in Section~\ref{secDescription}. We prove it correct and estimate its message complexity in Section~\ref{secCorrecness}. 
In Section~\ref{secExtensions}, we describe the algorithm's enhancements, such as splitting and merging coins and parallelizing transactions. We discuss how to improve \emph{TRAIL}'s efficiency by internally verifying transactions involving only a single shard. We also discuss practical implementation concerns such as recovering client information, bootstrapping the system, load balancing and shard maintenance.  In Section~\ref{secEvaluation}, we present performance evaluation of \emph{TRAIL} in an abstract simulator. We give an analysis of how \emph{TRAIL} improves network's meant time to failure. 
In Section~\ref{secRelated}, we describe related literature and discuss how various \emph{PBFT} enhancements and replacements 
can be used in \emph{TRAIL}, as well as how \emph{TRAIL} can be used to ensure shard failure tolerance in other sharding algorithms. We conclude the paper in Section~\ref{secEnd} by discussing potential further \emph{TRAIL} development directions.

\section{Network Model, Problem Statement, \emph{PBFT}}\label{secPreliminaries}
\noindent
\textbf{System model.}
We assume a peer-to-peer network. Each peer has a unique identifier. 
A peer may send a message to any other peer so long as it has the receiver's identifier.
Peers communicate through authenticated channels: the receiver of the message may always identify the sender. 
The communication channels are FIFO and reliable.

Network peers are grouped into \emph{shards}. Every shard has a unique identifier. For simplicity, we assume that all shards are the same size $s$. Each shard maintains a portion of the blockchain's data. Each shard peer stores a copy of its shard's data. Any peer may determine the shard identifier of any other peer in the network.

Peers are either correct or faulty. Faults are \emph{Byzantine}~\cite{lamport1982byzantine}: a faulty peer may behave arbitrarily. A \emph{peer tolerance threshold} $f$ is the maximum number of faulty peers that a shard can tolerate. A shard is correct if it has at most $f$ faulty peers.  The shard is faulty otherwise.

\ \\
\textbf{Data model and the problem.} 
A \emph{coin} is a unit of ownership whose transitions are 
recorded by the network.
Each coin has a unique identifier. A \emph{wallet} is a collection of coins. Each shard is responsible for storing and updating a disjoint subset of the network's wallets. A \emph{client} is an entity that owns a wallet. We assume a client is external to the peer-to-peer network but may communicate with any peer. Clients may submit transactions to the network requesting a coin to be moved from a \emph{source wallet} to a \emph{target wallet}. 
The peers are able to authenticate wallet owner; the peers accept transaction requests only from the source client. 
The approval of the target wallet owner is not required.

A cryptocurrency algorithm constructs a sequential ledger of transactions reflecting coin movements. Two transactions $t1$ and $t2$ are \emph{consequent} in this ledger if they operate on the same coin and there is no transaction $t3$ also operating on this coin such that $t3$ comes after $t1$ and before $t2$. 

An algorithm \emph{state} is an assignment of values to variables in all processes. 
Algorithm code contains a sequence of actions guarded by boolean guard predicates.
An action whose guard evaluates to \textbf{true} is \emph{enabled}. An algorithm \emph{computation} is a sequence of steps such that for each state $s_i$, the next state $s_{i+1}$ is obtained by executing an action enabled in $s_i$.

To make the \emph{TRAIL} correctness argument more rigorous, we formally state the problem that it solves. 

\begin{definition} An algorithm solves \emph{the Currency Transmission Problem} if it constructs a transaction ledger satisfying the following two properties:\\
\emph{ownership continuity} -- for any pair of consequent transactions $t1$ and $t2$, the target of $t1$ is the source of $t2$; \\
\emph{request satisfaction} -- if the owner requests a coin movement from its wallet, this request is eventually satisfied.
\end{definition}

Ownership continuity is a safety property that requires that a coin can only be moved out of a wallet once for each time it is moved into it. That is, it precludes double-spend attacks and disallows spending the money that the client does not have. The request satisfaction property guarantees liveness: the client request is eventually fulfilled.

\ \\
\textbf{\emph{PBFT}.} \emph{PBFT} is a Byzantine-robust consensus algorithm. Its tolerance threshold is 
$f = \lfloor (n-1)/3 \rfloor $,
where $n$ is the total number of peers in the system.

Peers communicate directly with each other via message broadcast. One of the peers is a \emph{leader}. 
The leader linearizes client requests. 
A period of single leader continuous operation is a \emph{view}. 
\emph{PBFT} is in \emph{normal operation} if the leader is correct. Normal operation has three phases: pre-prepare, prepare, and commit.

Once the leader receives a client transaction request, it assigns it a unique sequence number and starts the \emph{pre-prepare} phase by broadcasting a pre-prepare message containing the transaction and a sequence number to all peers. In the \emph{prepare} phase, each peer receives the pre-prepare message and broadcasts a prepare message containing the information that it received from the leader. If a peer receives $n-f-1$ prepare messages that match the initial pre-prepare, this peer is certain that correct peers agree on the same transaction. In this case, 
the peer starts the \emph{commit} phase by broadcasting a commit message.
Once a peer receives $n-f$ commit messages, normal \emph{PBFT} operation concludes and the peer informs the client that the transaction is committed. The transaction is confirmed once the client receives $f+1$ commits.

If the leader is faulty, the client requests may not be carried out. In this case, the client or the peers initiate a \emph{view change} to replace the leader. The view change process is designed to maintain transaction consistency through the transition: if a transaction is committed by a correct peer in the old view, the new leader submits it with the same sequence number so that the rest of the peers commit it in the new view.

If the number of faulty peers does not exceed the peer tolerance threshold $f$, \emph{PBFT} guarantees the following three properties: \emph{agreement} -- if a correct peer confirms a transaction, then every correct peer confirms this transaction; \emph{total order} -- if a correct peer confirms transaction $t1$ before transaction $t2$, then every correct peer confirms $t1$ before $t2$; and \emph{liveness} --  if a transaction is submitted to a sufficient number of correct peers, then it is confirmed by a correct peer. 
We assume these properties also apply to correct clients. 

The first two properties are satisfied regardless of the network synchrony. The liveness property is guaranteed only if the network is partially synchronous, meaning that the message transmission delay does not indefinitely grow without a bound.

\begin{figure}[htbp]
    \centering
    \includegraphics[width=\columnwidth]{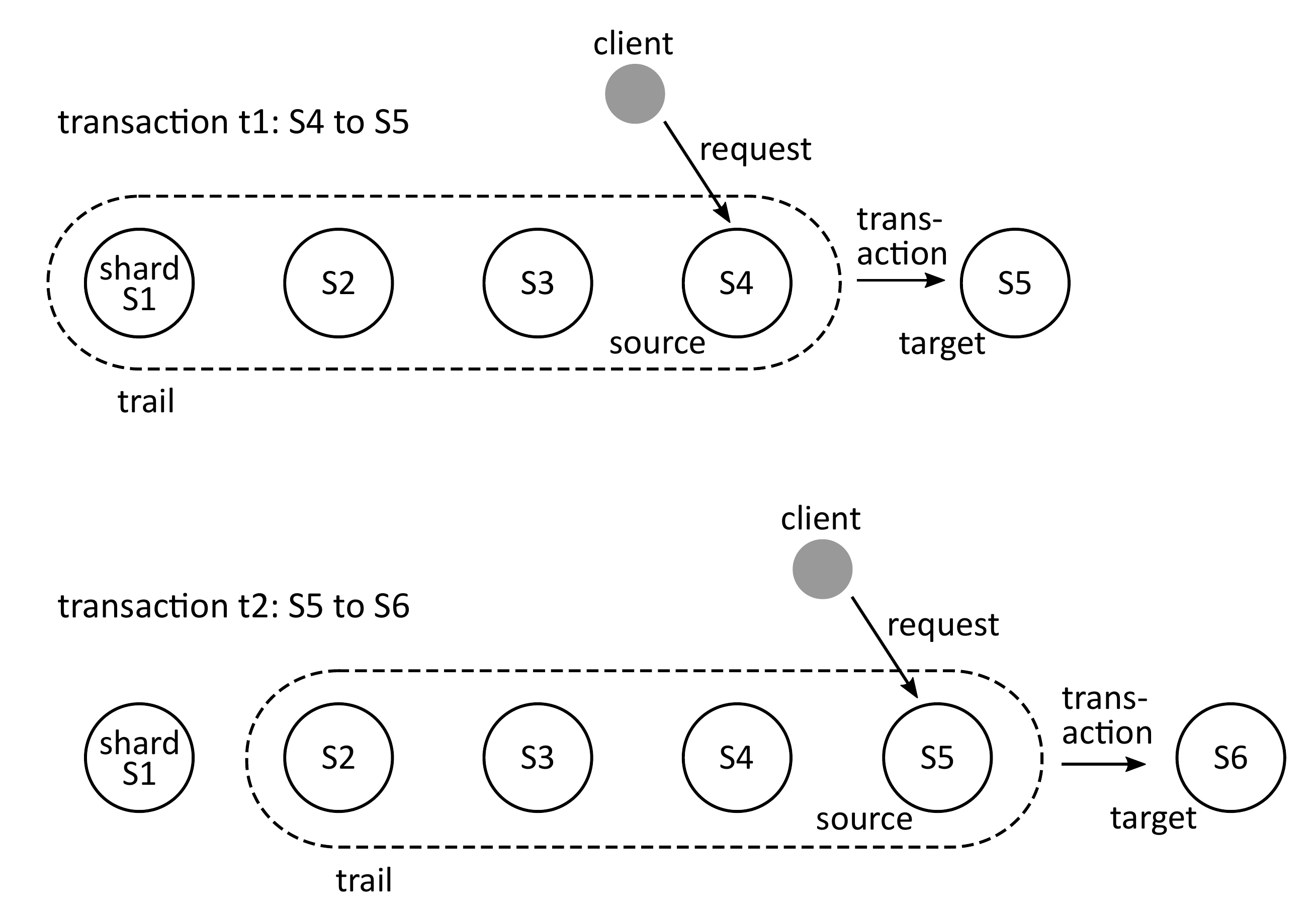}
    \caption{Trail membership modification under consequent transactions. The first transaction moves a coin from a wallet in shard $S4$ to a wallet in shard $S5$. The second moves the same coin from $S5$ to $S6$.}
    \label{figTrail}
\end{figure}

\begin{figure*}[htbp]
    \centering
    \includegraphics[height=5.1cm]{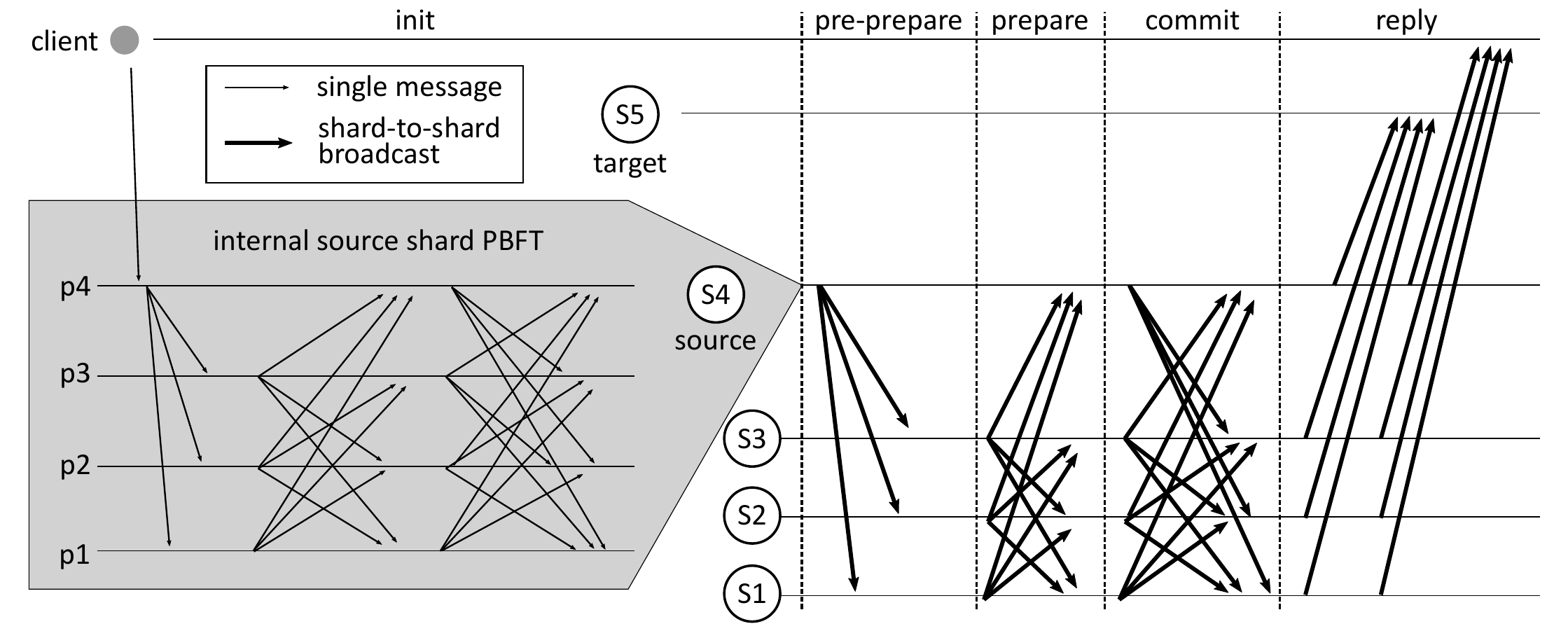}
    \caption{Message transmission in \emph{TRAIL}'s normal operation. The coin trail contains shards: $S1, S2, S3$ and $S4$. The coin is located in a wallet stored by shard $S4$. A client sends a transaction requesting to move the coin from the wallet of this source shard $S4$ to a wallet of  shard $S5$. First, the source shard runs internal \emph{PBFT}; then, it runs the phases of external shard \emph{PBFT}. After committing, the trail shards notify the client and the target shard.}
    \label{figTransmission}
\end{figure*}

\section{\emph{TRAIL} Description}\label{secDescription}

\noindent
\textbf{Algorithm outline.} The objective of the algorithm is to ensure the validity of coin transitions between wallets despite faulty peers and shards.

To counter malicious behavior of faulty shards, \emph{TRAIL} requires a collection of shards to agree on coin movement. This collection is called a \emph{trail}. A trail is composed of the $t$ unique shards whose wallets the coin visited most recently. Refer to Figure~\ref{figTrail} for an illustration. Notice that this trail is specific to a coin and changes as the coin moves from wallet to wallet. At any point in the computation, each coin may have its own separate trail of shards. We assume that each client knows the identities of the coins in its wallet as well as their trails.

\emph{Shard tolerance threshold} $F$ is the maximum number of faulty shards that \emph{TRAIL} may tolerate. The length of the trail, $t \geq 3F+1$.  
Note that $F$ and the peer tolerance threshold $f$ are not related. 

\emph{TRAIL} consists of two parts: (1) \emph{internal} source shard \emph{PBFT} and (2) \emph{external} trail \emph{PBFT}. To initiate the movement of a coin, the client that owns the source wallet sends a transaction request to the shard that holds it. To linearize received transaction requests and ensure that each individual transaction has an agreed-upon sequence number, the source shard peers execute internal \emph{PBFT}, see Figure~\ref{figTrail} for illustration, in which shard $S4$ is the source shard.

Once the source shard peers agree on this transaction, they initiate a modified external trail \emph{PBFT}. For that, each source shard peer broadcasts a \emph{pre-prepare} message to every peer of every trail shard. Once a trail peer receives $s-f$ 
such \emph{pre-prepare}s from the source shard, 
it initiates the next \emph{PBFT} phase by sending \emph{prepare} messages to every peer in every trail shard. In this way, each shard-to-shard broadcast emulates an individual message transmission in classic \emph{PBFT}. This continues until the external \emph{PBFT} instance commits. After that, each trail shard peer records the transaction in its ledger and notifies the target shard and the client. 

Once the target shard and the client are notified by $t-F$ trail shards, 
they record the transaction in their ledgers. The target shard, which is shard $S5$ in Figure~\ref{figTrail}, becomes the source shard for the next transaction. 

If the leader of the source shard, $S4$, is faulty, the other peers of the source shard execute a view change, switch to a new leader, and continue with internal \emph{PBFT}. 

If the source shard as whole is not faulty, i.e. the number of faulty peers in the source shard is below the tolerance threshold $f$, then the faulty peers may not influence trail shard. Indeed, for each external \emph{PBFT} message, the each peer of the trail shard expects at least $s-f$ individual messages.

If the number of peers in the source shard exceeds the tolerance threshold $F$ then the whole shard is faulty. In this case, the individual messages of the faulty source shard peers are equivalent to the faulty messages of the source shard. The external \emph{PBFT} guarantees that, despite the faulty shard, no spurious transactions will be recorded by the trail shards and that eventually the faulty leader shard is replaced. 

Specifically, the trail shards execute a view change, switch to a new shard as a leader, and continue with the consensus process, including a new internal \emph{PBFT} instance being performed within the new leader shard. Note that in the latter case, the record of the transaction may be placed in the trail shards but not in the faulty source shard that is nominally responsible for maintaining the source wallet record. This is an essential feature of our algorithm: the faulty source shard that stores the client wallet may be bypassed. 

Let us now present the algorithm in detail.

\ \\
\textbf{\emph{TRAIL} constants, variables, and functions.}  These constructs are shown in Algorithm~\ref{algVars}. 
Each peer with id $p$ knows the following constants: $f$ -- peer tolerance threshold (the maximum number of faulty peers in a correct shard); $F$ -- shard tolerance threshold (the maximum number of faulty shards); $s$ -- shard size; and $t$ -- trail size.

Several variables are common across transactions. We list them in a single place for convenience. Each transaction uses a coin identifier $coin$; a source wallet id $sWallet$; a target wallet is $tWallet$; a transaction sequence number $seq$ assigned by the source shard; and the sequence of shard ids $trail$ that indicates the trail shards for this coin at its present location.

Each peer maintains a $ledger$, which is a sequence of transaction records that it confirmed in a trail or received as a target. 

\emph{TRAIL} functions are shown in Algorithm~\ref{algFuncs}. They are grouped by their purpose.
Ledger maintenance functions are in Lines~\ref{lineLedgerStart}--\ref{lineLegerEnd}.
\emph{TRAIL} has two such functions. 

Function \textsc{Record} appends the transaction record to the ledger.
Function \textsc{IsPresent}($coin, wallet$) returns \textbf{true} if the ledger's most recent transaction record about $coin$ moved it to $wallet$, i.e. there is a transaction where $wallet$ is the target wallet and this record is not followed by a transaction moving $coin$ from $wallet$ to a different target wallet. 

\emph{TRAIL} uses several functions for wallet lookup and communications. They are shown in Lines~\ref{lineCommunicationsBegin}--\ref{lineCommunicationEnd}. 
Function~\textsc{GetShard}($wallet$) returns the id of the shard that stores $wallet$. We assume that every peer is able to identify which shard maintains each $wallet$.
%
%
%
Functions \textsc{Send} and \textsc{Receive} are single-message transmissions to the specified sender and receiver with straightforward functionality. In function \textsc{SendToShard}($shard, message$), the sender peer broadcasts a $message$ to all peers in $shard$. Function \textsc{ReceiveFromShard}($shard, message$) returns \textbf{true} once the peer receives $message$ from $t-F$ unique peers of $shard$.  

The internal source shard \emph{PBFT} is represented by two functions in \emph{TRAIL}. They are shown in Lines~\ref{linePBFTbegin}--\ref{linePBFTend}. Function \textsc{StartShardPBFT} initiates the \emph{PBFT} operation. The last function \textsc{CompleteShardPBFT} signifies that the internal \emph{PBFT} is completed and the peers assigned sequence number $seq$ to the transaction.

\begin{figure}[htbp]
\removelatexerror
\input{alg1_variables}
\end{figure}

\ \\
\textbf{\emph{TRAIL} phases.}  The actions for the algorithm are presented in Algorithm~\ref{algActions}.
We only show normal operation code for \emph{TRAIL}. View change code is added accordingly. Client and target code is not shown. See Figure~\ref{figTransmission} for the illustration of algorithm operation.

\emph{TRAIL} phases execute the internal source \emph{PBFT} and the external trail \emph{PBFT}.
\textbf{Phase 0: Init} (see Lines~\ref{lineInitStart}--\ref{lineInitEnd}) starts when a peer receives a transaction request from a client. If the peer contains the source wallet, i.e. it is the source shard for the transaction, the peer initiates internal shard \emph{PBFT}. After the source shard runs classic \emph{PBFT}, if the shard is not faulty, all the source shard peers agree on the transaction and its sequence number. The completion of internal \emph{PBFT} starts \textbf{Phase 1: Pre-prepare} (Lines~\ref{linePrePrepareStart} through~\ref{linePrePrepareEnd}). Each source shard peer sends a \emph{pre-prepare} message to a peer of every trail shard.

The receipt of $s-f$ messages from the source shard starts \textbf{Phase~2: Prepare} (Lines~\ref{linePepareStart}--\ref{linePrepareEnd}) in all the trail shards. Once a peer of the trail shard ascertains that the coin is present in the source wallet, i.e. the transaction is valid, the peer sends a $prepare$ message to all of the trail shards.

In \textbf{Phase 3: Commit} (Lines~\ref{lineCommitStart}--\ref{lineCommitEnd}), each trail peer assembles the \emph{prepare} messages. Variable $prepShards.coin.seq$ collects the identifiers of the shards from which this peer has received $s-f$ \emph{prepare} messages. If the number of these identifiers is $t-F$, the peer sends $commit$ message to all trail shards, signifying that it is ready to commit.

\textbf{Phase 4: Reply} (Lines~\ref{lineReplyStart}--\ref{lineReplyEnd}) is similar to the \textbf{Commit} phase. Once enough \emph{commit} messages from trail shards arrive, the peer records the committed transaction to its ledger and notifies the peers of the target shard and the client.

\begin{figure*}[htbp]
\removelatexerror
\hspace{-3mm}
\begin{tabular}{c@{\hspace{3mm}}c}
\begin{minipage}[t]{0.5\textwidth}
   \input{alg2_basic_functions}
   \vspace{-4mm}
\end{minipage}
&
\begin{minipage}[t]{0.5\textwidth}
   \input{alg3_actions}
   \vspace{-4mm}
\end{minipage}
\end{tabular}
\vspace{10mm}
\clearpage
\end{figure*}

\section{\emph{TRAIL} Correctness and Efficiency}
\label{secCorrecness}

\noindent
\textbf{Correctness proof.} Consider \emph{SimpleTRAIL}: a simplification of \emph{TRAIL} that operates as follows. For every shard in \emph{TRAIL}, there is only a single peer in \emph{SimpleTRAIL}. For simplicity, we still refer to these individual peers as shards. Once a client submits a request in \emph{SimpleTRAIL} to the source shard, that single peer immediately assigns it a sequence number and runs classic \emph{PBFT} with the trail shard peers, acting as the leader.

\begin{lemma}\label{lemSimpleTRAIL}
\emph{SimpleTRAIL} solves the Currency Transmission Problem with at most $F$ Byzantine shards.
\end{lemma}
\begin{proof} We show that \emph{SimpleTRAIL} satisfies the two properties of the Currency Transmission Problem: ownership continuity and request satisfaction.

Let us discuss ownership continuity first. It requires that for any two consequent transactions $t1$ and $t2$, the target of $t1$ is the source of $t2$.  

Let us consider transaction $t1$ and transaction $t1'$ that is consequent with it. Let $trail(t1)$ be the sequence of processes that confirm $t1$. 
Consequent transactions operate on the same coin. Since peers can authenticate wallet owners, the peers reject a transaction that moves an unavailable coin. Hence, the $t1'$ must be by either by the owner of the source or the target wallet of $t1$.

Let $t1'$ be submitted by the owner of the source wallet of $t1$. In this case, it has to be confirmed by the peers of $trail(t1)$. These peers run \emph{PBFT}. The agreement property of \emph{PBFT} ensures that the trail confirms the same transactions. The total order property of \emph{PBFT} states that two transactions $t1$ and $t1'$ are confirmed in the same order. That is, $t1'$ must be confirmed after $t1$. However, $t1'$ is supposedly consequent with $t1$. That is, it tries to move the coin that is already spent. Hence, if $t1'$ is submitted by the owner of the source wallet of $t1$, it is rejected. Let $t1'$ be submitted by the target of $t1$. That is $t1' = t2$. In this case, the trail, $trail(t2)$, for transaction $t2$ is updated. Specifically, one peer $px$ is added to $trail(t1)$ and peer $py$ is removed. 

The newly added peer $px$ may be correct or Byzantine. In case $px$ is correct, by the design of \emph{SimpleTRAIL} it receives  confirmation messages similar to the client that submitted $t1$. We assume that the properties of \emph{PBFT} apply to clients as well as to the peers. That is, $px$ is subject to the agreement property of \emph{PBFT}. This means that it confirms $t1$. Also, $px$ is subject to the total order property of \emph{PBFT}, which means that $px$ confirms $t2$ after $t1$. That is, in this case, \emph{SimpleTRAIL} records two consequent transactions $t1$ and $t2$ such that the target of $t1$ is the source of $t2$. In other words, the ownership continuity property is satisfied. 

Let us consider the case of faulty $px$. The total number of faulty shards in $tail(t2)$ does not exceed the fault tolerance threshold $F$. The remaining correct peers of $tail(t2)$ all belong to $tail(t1)$. These peers and the client execute a view-change procedure of \emph{PBFT}, replace the leader and confirm $t2$. Due to the total order property of \emph{PBFT}, $t2$ is confirmed after $t1$ which satisfies the continuity property of the Currency Transmission Problem. 

Similarly, the liveness property of \emph{PBFT} guarantees that the transactions moving the coin are eventually confirmed. Thus, the request satisfaction property of the Currency Transmission Problem is also satisfied. Hence the lemma. 
\end{proof}

Given a computation of \emph{TRAIL}, let us define an \emph{equivalent} computation of \emph{SimpleTRAIL}. If every correct peer in a shard executes a phase, for example Phase 0, in \emph{TRAIL}, then the single peer that corresponds to the shard in \emph{SimpleTRAIL} executes this action. If every correct peer in a shard sends a message to every peer in a shard, then the corresponding peer in \emph{SimpleTRAIL} sends this message to the single target peer. If a correct peer receives $s-f$ messages from the peers in a single shard in \emph{TRAIL}, then a single correct peer receives a message in \emph{SimpleTRAIL}. The transaction request that is broadcast from a client to peers of the same shard in \emph{TRAIL} is equivalent to a single message to \emph{SimpleTRAIL}. $s-f$ replies from the peers in the same shard to the client in \emph{TRAIL} are a single reply in \emph{SimpleTRAIL}.

\begin{lemma}\label{lemEquivTRAIL}
For every computation of \emph{TRAIL}, there is an equivalent computation of \emph{SimpleTRAIL}.
\end{lemma}

\begin{proof}(Outline) We prove the lemma by considering an arbitrary computation of \emph{TRAIL} and constructing an equivalent computation of \emph{SimpleTRAIL}. Let us consider a single client request. 
If the client sends a transaction request to the source shard peers and the request is received, the source shard executes the internal \emph{PBFT}. If the shard leader is faulty, the shard peers execute a view change. If not, they proceed straight to agreeing on the transaction sequence number. 

In either case, the correct peers execute \textsc{CompleteShardPBFT} which is equivalent to the start of pre-prepare phase in \emph{SimpleTRAIL}. The correct peers in the source shard then broadcast a pre-prepare message to the peers in every trail shard. That is, at least $s-f$ of these messages are received by every peer of the trail shard. This is equivalent to broadcasting a pre-prepare message in \emph{SimpleTRAIL}. The rest of the \emph{TRAIL} computation proceeds according to \emph{PBFT} and construction of the \emph{SimpleTRAIL} computation is similar. 

Let us now consider the case of faulty source shard. In this case, trail shards execute an external \emph{PBFT} view change. This operation maps directly to the equivalent \emph{SimpleTRAIL} computation. 
This construction of the equivalent computation for all client requests proceeds similarly. Hence the lemma.
\end{proof}

Lemma~\ref{lemSimpleTRAIL} shows that the non-sharded \emph{SimpleTRAIL} algorithm solves the Currency Transmission Problem. Lemma~\ref{lemEquivTRAIL} proves that every computation of the sharded algorithm \emph{TRAIL} is equivalent to a computation of \emph{SimpleTRAIL}. That is, every computation of \emph{TRAIL} satisfies the two properties of the problem. That is, \emph{TRAIL} also solves this problem. Hence the below theorem. 

\begin{theorem} Algorithm \emph{TRAIL} solves the Currency Transmission Problem with at most $F$ Byzantine shards and at most $f$ individual Byzantine faults in each correct shard. 
\end{theorem}

\ \\
\textbf{Message complexity.} If $n$ is the number of participants, the number of messages \emph{PBFT} requires to satisfy a request during normal operation is in $O(n^2)$. Therefore, the internal \emph{PBFT} of the source shard takes  $O(s^2)$ messages, 
where $s$ is the shard size.
Since it takes $s^2$ messages to transmit a single message of the emulated external \emph{PBFT}, the complete message complexity of \emph{TRAIL} is in $O(s^2 + s^2t^2)$ where $t$ is the trail size. If $s=3f+1$ and $t=3F+1$, then the complexity is in $O(f^2 + f^2 F^2)$. 

If view changes are considered, \emph{PBFT} requires $O(n^4)$ messages in the worst case. Hence, the worst case message complexity for \emph{TRAIL} is in $O(f^4 + f^4 F^4)$. 


\section{\emph{TRAIL} Algorithmic Extensions and Implementation Considerations}
\label{secExtensions}
\noindent
\textbf{Parallelizing transactions, splitting, merging and mining coins.} The same source node may run multiple external or internal transactions so long as they concern different coins. 

Multiple coins may be merged and a coin may be split: this is analogous to creating $1$ dollar out of $100$ cents or vice versa.  Both operations may be convenient to simplify transactions or make them more efficient. If a coin is split, all its portions inherit the old coins' trail. 
Coin merging is a bit more involved since the merging coins, even if they are located in the same wallet, may have different trails. 
To merge, the two coins are marked as merging and their movement transactions are executed jointly. The coins are finally merged once they travel together for the length of the trail. At this point, they share all shards of the trail. 

To create, or mine, a new coin, it needs to acquire a trail of length $t$. This may be accomplished by forming a committee of arbitrary $t$ shards and running a \emph{PBFT} on this committee to agree on the new coin's trail. 

\ \\
\textbf{Optimizing internal transaction validation.} To decrease message overhead, transactions are divided into internal and external. In an \emph{internal transaction} the source and target wallet are maintained by the same shard. To confirm this transaction, the source shard does not consult the trail shards; it runs internal \emph{PBFT}, and thus relies on the internal shard fault tolerance to maintain wallet integrity. External transactions are processed as usual: with external \emph{PBFT}.

The trade-off for this optimization is decreased shard fault tolerance: the trail shards are not aware of the source shard internal transactions. However, shard failure may be determined by a failure detector \cite{failureDetectors,weakestDetector,omegaImplementation,doudou1998muteness,kihlstrom2003byzantine,baldoni2003consensus}. Such a detector establishes a shard failure and notifies other shards. In the event of a detected shard failure, the trail shards perform a \emph{failed shard recovery procedure} to restore the integrity of the system: the wallets maintained by the failed shard are moved to other shards and their contents are restored to the last known external transaction. 
The clients have to re-submit internal transactions. 

\ \\
\textbf{Wallet location, client data recovery, shard maintenance.}
While coins move between shards, the wallets are assumed to be stationary. For quick shard lookup by the client, the wallet id might contain the shard number. Alternatively, the wallet-to-shard mapping may be recorded in the same or in a separate ledger. For efficiency, wallets frequently participating in joint transactions may be moved to the same shard.

If a client loses its local information about the coin contents of its wallet, it may be able to recover it by conducting a network-wide query.  Note that asking the shard that keeps the wallet information alone is not sufficient: the shard may be faulty. Instead, the complete network broadcast is required. The trail shards that confirmed moving the coin should answer to the recovering client. Again, since some shards may be faulty, the client considers the coin present in its wallet if the whole trail confirms its location. 

In \emph{TRAIL} description, we assumed that the shard sizes are uniform. However, this does not have to be the case. Instead, the shards may grow and shrink as peers join or leave system. Shard sizes may also be adjusted in response to transaction load requirements. Shard membership may be maintained in the shard ledger or, alternatively, in a separate membership ledger.  

\ \\
\textbf{Algorithm parameter selection.} \emph{TRAIL} operates correctly regardless of the concrete values of shard size $s$ and trail size $t$. These parameters, however, affect the algorithm performance. Larger $s$ makes it less likely that the complete shard fails. Yet, larger $s$ makes the internal consensus algorithm less efficient. The smaller $s$ necessitates larger trail size $t$ to protect against shard failure. 

\section{Performance Evaluation and \\ Security Analysis}
\label{secEvaluation}
\noindent\textbf{Simulation setup.}
We evaluate the performance of \emph{TRAIL} in an abstract algorithm simulator QUANTAS~\cite{quantas}. QUANTAS simulates multi-process computation, message transmission and has extensive experimental setup capabilities. The simulator is implemented in C++.  It is optimized for multi-threaded large scale simulations~\cite{Shoshany2021_ThreadPool}.
The code for our \emph{TRAIL} implementation in QUANTAS as well as our performance evaluation data is available online~\cite{TrailGitHub,TrailData}.

The simulated network consists of individual peers. Each pair of peers is connected by a message-passing channel. Channels are FIFO and reliable.  A computation is modeled as a sequence of rounds. In each round, a peer receives messages that were sent to it, performs local computation, and sends messages to other peers.

Peers are divided into shards. Shard leaders propose transactions; clients are not explicitly simulated. A transaction has a $25\%$ probability of having source and target wallets in separate shards. Internal transactions are not externally verified by the trail of shards. If a shard is Byzantine, it generates invalid cross-shard transactions only. Specifically, the shard creates transactions moving coins that it has already spent.

\begin{figure}
\includegraphics[width=\columnwidth]{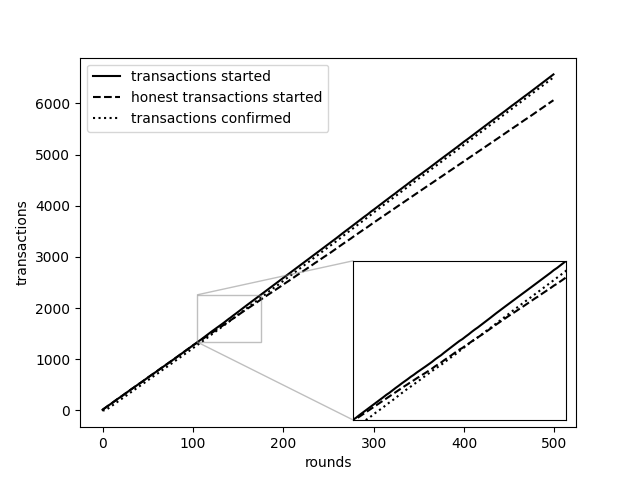}
\caption{Transactions approved over time without \emph{TRAIL} shard validation. The network approves both honest and malicious transactions.}

\label{novalidation-txs}
\end{figure}

\begin{figure}[htbp]
\includegraphics[width=\columnwidth]{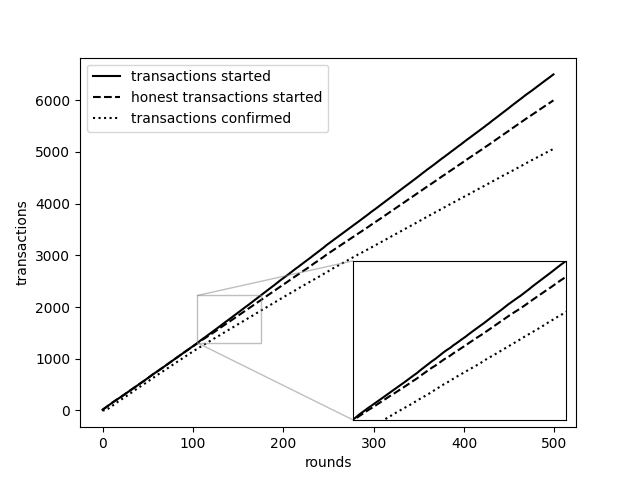}
\caption{Transactions approved over time with \emph{TRAIL}  shard validation.  The network approves honest transactions only.}

\label{normal-txs}
\end{figure}

\begin{figure}[htbp]
\includegraphics[width=\columnwidth]{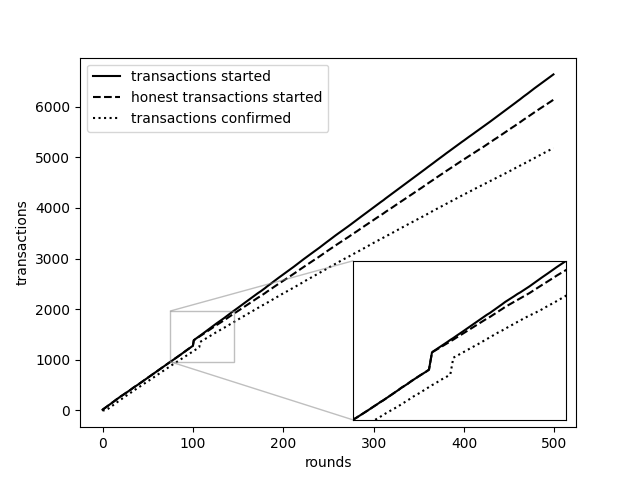}
\caption{Transactions approved over time in \emph{TRAIL} with shard validation and wallet recovery from the failed shards.  Correct shards detect the failure and submit additional transactions moving coins from the failed shards. }

\label{rollback-txs}
\end{figure}

\begin{figure}
\includegraphics[width=\columnwidth]{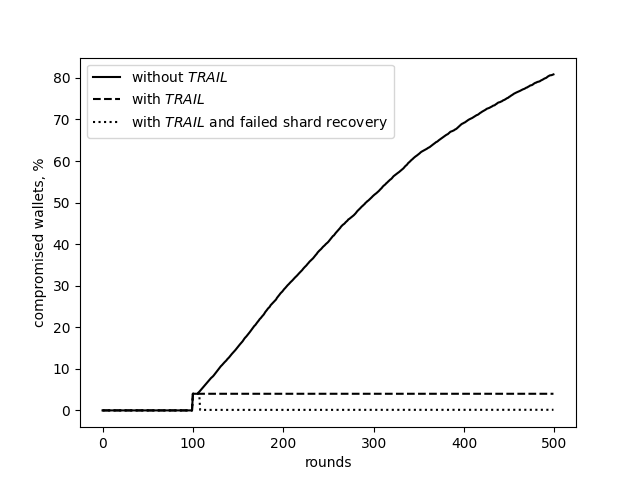}
\caption{Percentage of wallets compromised by malicious transactions.}

\label{corrupt-wallets}
\end{figure}

\begin{figure}
\includegraphics[width=\columnwidth]{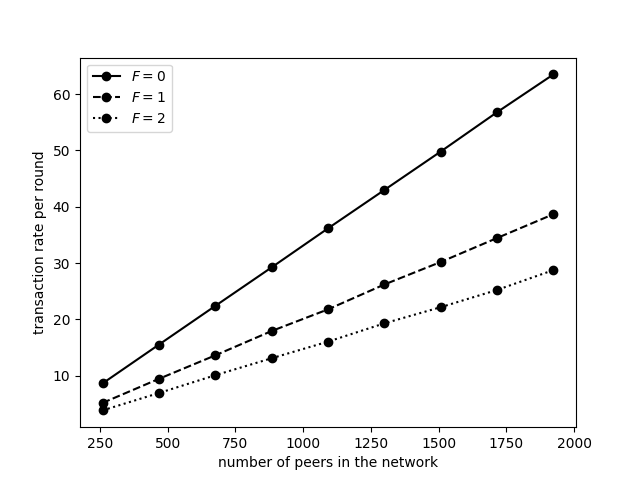}
\caption{Throughput with respect to the number of peers in the network for different fault tolerance levels.}

\label{throughput-scaling}
\end{figure}

\ \\
\noindent\textbf{Experiment description.} Figures~\ref{novalidation-txs}, \ref{normal-txs}, ~\ref{rollback-txs}, and \ref{corrupt-wallets} show the dynamics of transaction processing during a computation. In these simulations, a computation runs for $500$ rounds. The internal faulty peer tolerance threshold $f$ is $7$. The shard size is $s=3\cdot f+1 = 22$. The faulty shard tolerance threshold $F$ is $2$. This makes the trail size $3 \cdot F+1= 7$. 
In the experiments, the number of actual faculty shards is equal to the shard fault tolerance threshold $F$; that is, we run the experiments with maximum tolerance. The faulty shards behave correctly at the start of the simulation and fail at round $100$. The total number of shards in the system is $S=50$. Therefore, the network size is $S\cdot s = 1100$.
We run $15$ experiments per data point and show the average of the results.
Figures~\ref{novalidation-txs},~\ref{normal-txs}, and~\ref{rollback-txs} show the accumulated counts of started and confirmed transactions. We distinguish between the \emph{honest} transactions generated and the total number of transactions, which includes \emph{malicious} transactions generated by faulty shards. The number of honest confirmed transaction is lower.

In Figure~\ref{novalidation-txs}, no cross-shard validation is performed. In this figure, the number of confirmed transactions matches the total number of transactions; that is, transactions are confirmed whether they are malicious or not. 
The graph indicates a certain delay before transaction starting and confirmation due to the operation of external and internal \emph{PBFT}. 

In Figure~\ref{normal-txs}, \emph{TRAIL} validates the external transactions. Malicious transactions are not confirmed, and the total number of confirmed transactions only accounts for the honest transactions. 

In graph shown in Figure~\ref{rollback-txs}, \emph{TRAIL} uses the failed shard recovery procedure. Specifically, at round $100$, when $F$ shards fail, 
\emph{TRAIL} detects the faults and generates transactions to move the coins from the faulty shard wallets to the correct ones. This explains the increase in the transaction generation and confirmation rates near round $100$ in the figure. 

Figure~\ref{corrupt-wallets} shows the effect of malicious transactions on the overall system integrity. A wallet is \emph{compromised} if it is in a faulty shard or if it receives a coin from a compromised wallet that is not possessed by that compromised sender wallet. A wallet is \emph{safe} otherwise. The safety of a compromised wallet can be restored by the failed shard recovery procedure. The solid line in Figure~\ref{corrupt-wallets} shows the wallet compromise trend if no cross-shard validation is used. In this case, the failed shards continuously generate malicious transactions, compromising progressively larger number of wallets in the correct shards of the network. In the case in which cross-shard validation is used, the dashed line in Figure~\ref{corrupt-wallets}, the number of compromised wallets does not exceed the number of wallets in the failed shard. In the case with the failed shared recovery procedure, all wallets are eventually become safe again: the correct shards generate coin wallet recovery transactions. The delay in wallet recovery shown in the graph is due to the validation of these transactions.

Figure~\ref{throughput-scaling} shows the performance of \emph{TRAIL} at scale. We run these experiments with a maximum of $148$ shards made up of $13$ peers per shard. Each data point represents the average throughput from 5 simulations of 200 rounds each. We plot \emph{TRAIL}'s performance with three shard tolerance thresholds $F$: $0$, $1$ and $2$. There is no cross-shard validation in case of $F=0$. The figure indicates that the performance of \emph{TRAIL} scales well with network size increase. Larger fault tolerance thresholds incur more overhead. Therefore, the transaction rate is lower for higher values of $F$.

\begin{figure}[htbp]
\includegraphics[width=\columnwidth]{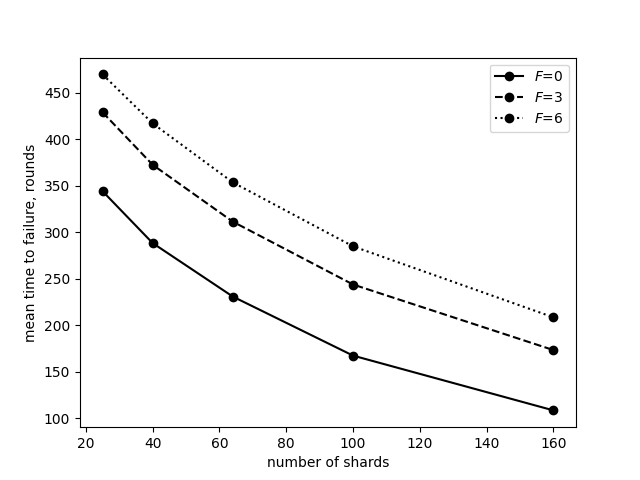}
\caption{The mean time to failure with respect to the number of shards.}
\label{MTTF}
\end{figure}

\begin{figure}[htbp]
\includegraphics[width=\columnwidth]{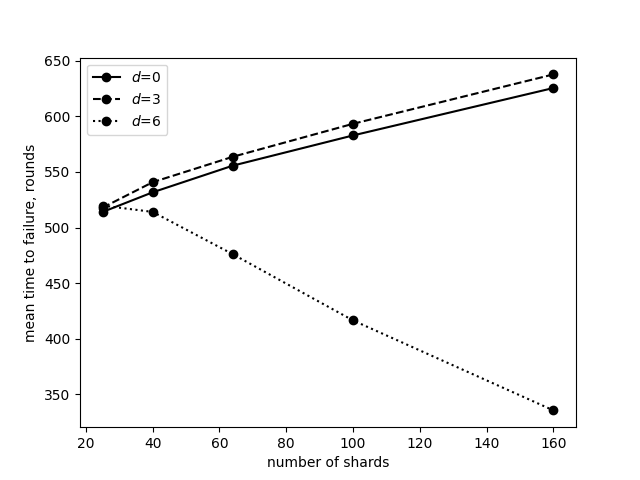}
\caption{The mean time to failure with individual failed shard detection and removal. System fault tolerance threshold $F =3$, detection delay $d$.}
\label{MTTF_with_rollbacks}
\end{figure}

\ \\
\textbf{Security analysis.}
We investigate the mean time to failure (MTTF) of \emph{TRAIL}. The system fails when the number of faulty shards exceeds its shard tolerance threshold $F$. In every round, a single peer fails uniformly at random. We run the experiment until the system fails. We carry out $100$ experiments per data point. The total number of peers is fixed at $1600$. We vary the number of shards and, hence, the number of peers in each shard.

In Figure~\ref{MTTF}, we show the MTTF for three different thresholds. If $F = 0$, there is no cross-shard validation and the system fails with its first failed shard. The graph shows that as the number of shards increases, i.e. as the number of peers in a shard decreases, the MTTF goes down. For a fixed number of shards, the MTTF is greater for a higher tolerance threshold since it requires more shards to fail before the system fails. 

In Figure~\ref{MTTF_with_rollbacks}, we show \emph{TRAIL} performance with failed shard detection and the failed shard recovery procedure. A detector discovers individual shard failure. Once discovered, the failed shard is removed from the system. The detection may not be immediate. In the figure, we plot three different delays $d$ for this detector. If $d =0$, the shard failure detection is instantaneous. For this experiment the shard failure threshold is fixed at $F=3$. 

The dynamics of the system are remarkable. The MTTF for both $d=0$ and $d=3$ trends upward as the number of shards increase, i.e. as the shard size decreases. The detection removes failing shards fast enough to counter accumulated peer failures and delay the overall system failure. 
Interestingly, the detector delay of $3$ performs better than $0$. In case of non-zero delay, the failed shard remains in the system until detection and it has the opportunity to accumulate more failed peers which are then removed with the shard. In case $d=6$, the detection delay is so large that the accumulated peer faults result in failed shards and the overall system failure. Therefore, the MTTF for $d=6$ trends downward. 

Our system security analysis indicates the increased resilience of the system with cross-shard validation of \emph{TRAIL}.

\section{Related Work and Its Application to \emph{TRAIL}}
\label{secRelated}

\noindent
\textbf{Sharding blockchains.} A number of sharding cryptocurrency blockchains are presented in the literature~\cite{rscoin,omniledger,rapidchain,ostraka,chainspace,elastico}. See Le {\em et al.}~\cite{shardingSurvey} for an extensive recent survey. 
We, however, have not seen an approach where sharding is done on the basis of the coin trail. 

We are not aware of any blockchain that is robust to shard failure. 
However, some blockchains attempt to mitigate it. 
\emph{RSCoin}~\cite{rscoin} requires central bank involvement. 
\emph{Omniledger}~\cite{omniledger}, \emph{Chainspace}~\cite{chainspace}, and \emph{Ostraka}~\cite{ostraka} use optimistic single-shard transaction commits with post factum detection of inconsistent transactions inserted by malicious shards. \emph{Rapidchain}~\cite{rapidchain} uses the Cuckoo rule~\cite{awerbuch2006towards,sen2012commensal}, which reshuffles the shard membership when new nodes join the shard in an attempt to foil the adversary from packing a shard with faulty nodes and exceeding its tolerance threshold.

We believe that most of the published blockchains, even if they do not use \emph{PBFT}, can employ \emph{TRAIL} to fortify themselves against shard failures. For this, consensus on transactions has to be deferred until the transaction's trail confirms it.

\ \\
\textbf{\emph{PBFT} optimizations and replacements.}
There are numerous proposals to optimize \emph{PBFT} performance.  See Wang {\em et al.} for a survey~\cite{wang2022bft}. Several propose using multiple leaders concurrently~\cite{mirbft,bigbft,rcc}. \emph{Mir-BFT}~\cite{mirbft} and \emph{RCC}~\cite{rcc} suggest accelerating \emph{PBFT} by processing non-conflicting requests concurrently. The algorithms have multiple leaders that process these requests simultaneously. \emph{BigBFT}~\cite{bigbft} further enhances parallelism by pipelining subsequent requests. 

Another approach is to shift some processing to the client. This eliminates the communication between peers and lowers the message complexity to $O(n)$. For example, \emph{QU}~\cite{qu} requires the client to directly communicate with all peers and validate their responses. \emph{QU} needs $5f+1$ peers for correctness. \emph{HQ}~\cite{hq} improves this approach and lowers the number of required peers to the theoretical minimum of $3f+1$. A variation of this scheme is \emph{Zyzzyva}~\cite{zyzzyva}, where the leader sends a message to other peers. The peers then reply to the clients directly rather than communicating amongst themselves. \emph{Zyzzyva} message complexity is in $O(n)$ during normal operation and in $O(n^2)$ during leader change. \emph{SBFT}~\cite{sbft} similarly optimizes message complexity by having designated collector peers, rather than the client, collect other 
peers' messages.
\emph{Hotstuff}~\cite{hotstuff} optimizes \emph{PBFT} by routing messages of all phases of \emph{PBFT} through the leader. This  decreases the message complexity to $O(n)$ while increasing the number of sequential consensus rounds. 

The above \emph{PBFT} optimizations can be applied in \emph{TRAIL} to the internal shard consensus protocol in a straightforward manner. Most of these optimizations can also be applied to the external \emph{TRAIL} algorithms as well. 
%

\section{Future Work}\label{secEnd}

The \emph{TRAIL} algorithm presented in this paper is the first to systematically address Byzantine shard failure protection in blockchains for cryptocurrencies. We foresee that it might be developed into a fully-fledged system. Alternatively, \emph{TRAIL} may be used as an add-on component to fortify existing blockchains against shard failure. 
As a third alternative, \emph{TRAIL} may be enhanced to handle more challenging conditions, such as network partitioning~\cite{hood2021partitionable} or dynamic networks~\cite{bricker2022blockchain}. Any and all of these alternative directions will increase the robustness of future blockchains.

\bibliographystyle{IEEEtran}
\bibliography{eyewitness}

\end{document}